\documentclass[journal]{IEEEtran}

\ifCLASSINFOpdf
 
\else
 
\fi

\usepackage{float}
\usepackage[colorlinks,linkcolor=blue]{hyperref}
\usepackage[latin1]{inputenc}
\usepackage{amsmath}
\usepackage{graphicx}
\usepackage{amssymb}
\usepackage{amsthm}
\usepackage{verbatim}
\usepackage{epsfig}
\usepackage[labelfont=bf]{caption}
\usepackage{subcaption}
\usepackage{enumerate}

\usepackage{amssymb}
\usepackage{amsthm}
\usepackage{amsmath}
\usepackage{amsmath,bm}
\usepackage{mathrsfs}
\usepackage{lineno,hyperref}
\usepackage{dsfont}
\usepackage{mathtools}

\newtheorem{definition}{Definition}
\newtheorem{assumption}{Assumption}
\newtheorem{lemma}{Lemma}
\newtheorem{remark}{Remark}

\newtheorem{proposition}{Proposition}

\def\begquo{\begin{quote}}
	\def\endquo{\end{quote}}
\def\begequarr{\begin{eqnarray}}
\def\endequarr{\end{eqnarray}}
\def\begequarrs{\begin{eqnarray*}}
	\def\endequarrs{\end{eqnarray*}}
\def\begarr{\begin{array}}
	\def\endarr{\end{array}}
\def\begequ{\begin{equation}}
\def\endequ{\end{equation}}
\def\lab{\label}
\def\begdes{\begin{description}}
	\def\enddes{\end{description}}
\def\begenu{\begin{enumerate}}
	\def\begite{\begin{itemize}}
		\def\endite{\end{itemize}}
	\def\endenu{\end{enumerate}}

\def\lef[{\left[\begin{array}}
	\def\rig]{\end{array}\right]}

\def\begcen{\begin{center}}
	\def\endcen{\end{center}}
\def\begrem{\begin{remark}\rm}
	\def\endrem{\end{remark}}
\def\begdef{\begin{definition}}
	\def\enddef{\end{definition}}
\def\begpro{\begin{propositionosition}}
	\def\endpro{\end{propositionosition}}
\def\begfac{\begin{fact}}
	\def\endfac{\end{fact}}
\def\begass{\begin{assumptionption}}
	\def\endass{\end{assumptionption}}
\def\begsubequ{\begin{subequations}}
	\def\endsubequ{\end{subequations}}

\def\begmat#1{\begin{bmatrix}#1\end{bmatrix}}
\def\begali#1{\begin{align}{#1}\end{align}}
\def\begalis#1{\begin{align*}{#1}\end{align*}}

\def\caly{{\cal Y}}

\def\calh{{\cal H}}

\def\calt{{\cal T}}
\def\calb{{\cal B}}
\def\calm{{\cal M}}

\def\call{{\cal L}}

\def\cala{{\cal A}}

\def\liminf{\lim_{t \to \infty}}

\def\L2e{{\cal L}_{2e}}

\def\rea{\mathbb{R}}

\def\adj{\mbox{adj}}
\def\col{\mbox{col}}

\def\TAC{{\it IEEE Trans. Automatic Control}}

\def\AUT{{\it Automatica}}

\def\CST{{\it IEEE Trans. Control Systems Technology}}

\def\TPS{{\it IEEE Trans. on Power Systems}}

\usepackage{color}

\begin{document}

\title{State Observation of Power Systems Equipped with Phasor Measurement Units: The Case of Fourth Order Flux-Decay Model}

\author{Alexey~Bobtsov,~\IEEEmembership{Senior Member,~IEEE,}
        Romeo~Ortega,~\IEEEmembership{Life Fellow,~IEEE,} Nikolay~Nikolaev,~\IEEEmembership{Member,~IEEE,} Johannes~Schiffer,~\IEEEmembership{Member,~IEEE,}
        and~M.~Nicolai~L.~Lorenz-Meyer
\thanks{A. A. Bobtsov and N. Nikolaev are with ITMO University, 197101, Saint
	Petersburg, Russian Federation (e-mail: bobtsov@mail.ru, nanikolaev@itmo.ru)}
\thanks{R. Ortega is with CNRS L2S, 91192 Gif-sur-Yvette, France (e-mail:
	ortega@lss.supelec.fr) and with ITMO University}
\thanks{M. N. L. Lorenz-Meyer and J. Schiffer are with Brandenburg University of
	Technology Cottbus-Senftenberg, 03046 Cottbus, Germany (e-mail: florenzmeyer, schifferg@b-tu.de)}

\thanks{Manuscript received March 19, 2020.}}

\maketitle

\begin{abstract}
The problem of effective use of Phasor Measurement Units (PMUs) to enhance power systems awareness and security is a topic of key interest. The central question to solve is how to use this new measurements to reconstruct the {\em state} of the system. In this paper we provide the first solution to the problem of (globally convergent) state estimation of multimachine power systems equipped with PMUs and described by the {\em fourth order} flux-decay model. This work is a significant extension of our previous result, where this problem was solved for the simpler third order model, for which it is possible to recover {\em algebraically} part of the unknown state. Unfortunately, this property is lost in the more accurate fourth order model, significantly complicating the state observation task. The design of the observer relies on two recent developments proposed by the authors, a parameter estimation based approach to the problem of state estimation and the use of the Dynamic Regressor Extension and Mixing (DREM) technique to estimate these parameters. The use of DREM allows us to overcome the problem of lack of persistent excitation that stymies the application of standard parameter estimation designs. Simulation results illustrate the latter fact and show the improved performance of the proposed observer with respect to a locally stable gradient-descent based observer.
\end{abstract}

.
\IEEEpeerreviewmaketitle

\section{Introduction}
\label{sec1}

\IEEEPARstart{P}{ower} systems are experiencing major changes and challenges, such as an increasing amount of power-electronics-interfaced equipment, growing transit power flows and fluctuating (renewable) generation, see \cite{ winter15}. Therefore power systems are operated under more and more stressed conditions and, thus, closer to their stability limits as ever before. In addition, as detailed in \cite{ milano18}, their dynamics become faster, more uncertain and also more volatile. Hence, fast and accurate monitoring of the system states is crucial in order to ensure a stable and reliable system operation, see \cite{ZHAetal}. This, however, implies that the conventional monitoring approaches based on steady-state assumptions are no longer appropriate and instead novel {\em state observers}\footnote{In the power systems community, to distinguish it from the steady-sate case, the qualifier``dynamic" is added to the state observation problem, and it is sometimes called ``dynamic state estimation".} tools have to be developed, see \cite{ZHAetal,SINPAL}.

By recognizing this need, the design of state observers has become a very active research area in the past years. The interest has been further accelerated by the growing deployment of PMUs, see \cite{ terzija10}. The vast majority of the reported results on this matter rely on the use of {\em linear} systems-based theories, {\em e.g.}, the use of Kalman filters, whose performance is assessed only via simulations, see \cite{GHAKAM,PAUetal,WECetal,ZHAetal}. As thoroughly discussed in \cite{LORetal}. this approach suffers from several major drawbacks. 

Recently \cite{LORetal}, the authors provided a globally convergent solution to the state observation problem for the case when the generators are modelled by the classical {\em third order} flux-decay model. Instrumental for the solution of the problem was the observation that, for this model, it is possible to recover {\em algebraically} part of the unknown state. It is widely recognized \cite{MACetal,SINPAL,WECetal} that to improve the precision of the model, it is necessary to include additional dynamic effects, leading to a fourth order model. Unfortunately, for this case, the algebraic reconstruction of part of the state is {\em impossible}, significantly complicating the state observation task. 

In this paper we provide the first solution to the state observation problem for multimachine power systems described by the fourth order model. The design of the observer relies on two recent developments proposed by the authors, a generalization of the Parameter Estimation-based Observer (GPEBO) \cite{ORTetalscl}, and the use of the DREM technique \cite{ARAetaltac} to estimate these parameters. GPEBO was used in \cite{ORTetalgpebo} for the design of observers for bio-chemical reactors and the simplest problem of state estimation of third-order power systems with {\em measurement} of the rotor angle. The latter, practically restrictive assumption, is removed here significantly widening the applicability of the result. Thanks to the use of DREM it is possible to overcome the problem of lack of persistent excitation that stymies the application of standard observer designs. Simulation results, presented in Section \ref{sec6}, illustrate  the latter fact and shows the improved performance of the proposed observer with respect to a locally stable gradient-descent based observer.

Throughout the paper we restrict ourselves to the study of a single generator. As shown in \cite{LORetal}, thanks to the incorporation of the PMUs, for the purposes of observer design it is possible to treat multimachine systems as a set of decentralized single machines, hence our result can be extended in a straightforward way to the multimachine case. In the interested of brevity we omit the details of this generalization, and refer the interest reader to \cite[Section II]{LORetal} for the details.    

\section{Mathematical Model and Problem Formulation}
\label{sec2}

We consider the well-known fourth-order model of the single machine system given by  \cite[eq. (1)]{GHAKAM}, see also \cite[Chapter 5.4]{SAUetal} and \cite[Chapter 11.1.7.1]{MACetal} , 
\begsubequ
\lab{x}
\begali{
	\label{x1}
	\dot x_1&= x_2\\
	\label{x2}
	\dot x_2&=-a_0x_2 + b_0(u_1 -y_5)\\
	\label{x3}
	\dot x_3&=-a_2 x_3 + b_2 y_2\sin(x_1-y_1)\\
	\label{x4}
	\dot x_4&=-a_1 x_4+b_1 y_2\cos(x_1-y_1)+c_1u_2,
}
\endsubequ

where the {\em unknown} state and input variables are defined as
\begalis{
	x&:=\begmat{ x_1 & x_2 & x_3 & x_4\end{bmatrix}^\top=\begin{bmatrix} \delta & \omega & E_d' & E_q'}^\top\\
	u&:=\begmat{ u_1 & u_2\end{bmatrix}^\top=\begin{bmatrix} P_m & E_f }^\top,
}
with $\delta$ the rotor angle, $\omega$ the shaft speed, $E_d'$ and $E_q'$ the direct and quadrature axis internal voltages, respectively, $P_m$ the mechanical power and $E_f$ the field voltage, and we defined the positive constants
\begalis{
	a_0&:=\frac{\omega_0D}{2H},\;b_0:=\frac{\omega_0}{2H}\\
	a_1&:=\frac{1}{T_{d0'}}\frac{x_d}{x_d'},\;b_1:=\frac{1}{T_{d0'}}\frac{(x_d-x_d')}{x_d'},\;c_1:=\frac{1}{T_{d0'}}\\
	a_2&:=\frac{1}{T_{q0'}}\frac{x_q}{x_q'},\;b_2:=\frac{1}{T_{q0'}}\frac{(x_q-x_q')}{x_q'}.
}
All variables and constants are defined in Table 1 in the Appendix. The PMU {\em measurements} are

\begin{align}
\lab{y}
y:=\begin{bmatrix} y_1 & y_2 & y_3 & y_4 & y_5 & y_6\end{bmatrix}^\top=\\
\nonumber
\begin{bmatrix} \theta_t & V_t & \phi_t & I_t & P_t & Q_t\end{bmatrix}^\top,
\end{align}
where 
\begsubequ
\lab{y456}
\begali{
	\lab{y4}
	y^2_4&=\frac{1}{x_q'^2} (x^2_3+x^2_4+y_2^2\\
	\nonumber
	&-2y_2[x_4 \cos(x_1-y_1)+x_3 \sin (x_1-y_1))] \\
	\lab{y5}
	y_5&=\frac{y_2}{x_q'}[x_4  \sin(x_1-y_1) -x_3 \cos(x_1-y_1)] \\
	y_6&=\frac{y_2}{x_q'}[x_4  \cos(x_1-y_1) +x_3 \sin(x_1-y_1)-y_2]
	\lab{y6}
}
\endsubequ

We underscore the presence of the terminal bus voltage $y_2=V_t$ that, in a multimachine scenario, captures the effect of the interconnection among the machines, see \cite[Section II]{LORetal} for details. \\

To formulate the observer problem we  make the following assumptions on systems prior knowledge  and the available {measurements}.

\begin{assumption}
	\label{ass1} \em
	The signals $u$ are {\em measurable} and the electrical subsystem parameters $(a_1,a_2,b_1,b_2.c_1)$ are {\em known}.\footnote{As usual in observer problems \cite{BER}  we assume that $u$ is bounded and such that all state trajectories are {\em bounded.}}
\end{assumption}

\begdes
\item {\bf Problem Formulation:} Consider the SMIB power system (\ref{x}) with measurable outputs \eqref{y}, \eqref{y456}, verifying Assumption \ref{ass1}. Design an observer
\begalis{
	\dot \chi &= F(\chi,u,y)\\
	\hat x&:=H(\chi,u,y)
}
such that $\liminf \tilde x(t)=0$, where we, generically, define the estimation {\em errors} $\tilde{(\cdot)}:=\hat{(\cdot)}-(\cdot)$. 
\enddes

\begrem
\lab{rem1}
The model \eqref{x1}-\eqref{x4} is obtained making the standard assumption that the {\em stator resistance is zero}.  Moreover, the expressions in \eqref{y5} and \eqref{y6} are obtained by neglecting transient saliency, thus assuming that the direct- and quadrature-axis transient reactances, $x_d'$ and $x_q'$, respectively, are {\em equal}.
\endrem

\begrem
\lab{rem2}
The expression for $y_4$ given in \eqref{y4} can be derived using the direct- and quadrature-axis currents $I_d$ and $I_q$, respectively, defined as 
\begalis{
	I_d&:=\frac{1}{x_d'}[x_4 -y_2 \cos(x_1-y_1)]\\
	I_q&:=\frac{1}{x_q'}[-x_3+y_2  \sin(x_1-y_1)],
}  
with $x_d'=x_q'.$ 
\endrem
%

\section{3rd and 4th Order Models: A Fundamental Difference}
\lab{sec3}

As indicated in Section \ref{sec1}, in \cite{LORetal} we present a globally convergent solution to the state observation problem for the case when the generators are modeled by the classical {\em third order} flux-decay model given by\footnote{To avoid cluttering, and with some abuse of notation, we keep the same symbols for the 3rd and the 4th order models.}
\begalis{
	\dot x_1&=x_2\\
	\dot x_2&=-a_1 x_2 + a_2[P_m - Y V x_3 \sin(x_1)]\\
	\dot x_3 &= -a_3 x_3+ a_4 V \cos(x_1)+ E_f,
}
where the {\em unknown} state is defined as

\begalis{
	x&:=\begmat{ x_1 & x_2 & x_3 \end{bmatrix}^\top=\begin{bmatrix} \delta-\theta_t & \omega & E_q' }^\top, }

with $Y>0$ is the susceptance of the network admittance, and $V$ is the terminal voltage. All other parameters of the model $a_i,\;i=1,4$, are constant. The {\em measurements}, which are provided via PMUs, are defined as 
\begali{
	\nonumber
	y_1&=V \\
	\nonumber
	y_2&=Y V x_3 \sin(x_1)\\
	\nonumber
	y_3&= Y Vx_3\cos(x_1) - Y V^2 \\
	\lab{y3rd}
	y_4^2 &= Y^2 [x^2_3+V^2-2Vx_3 \cos(x_1)],  \\
	y_5 &= f_{t},
}  
where $y_1>0$ is the terminal voltage, $y_2$ is the active power, $y_3$ is the reactive power, $y_4$ is the terminal current and $y_{5}$ the terminal voltage frequency. 

In \cite{LORetal} it was shown that it is possible to algebraically reconstruct the states $x_1$ and $x_3$ from the measurements $y$ as follows.  

\begin{proposition}\em \cite{LORetal}
	\lab{pro1}
	The states $x_{1}$ and $x_{3}$ can be determined uniquely from the PMU measurements \eqref{y3rd} via
	\begin{align*}
	x_{3} &= \sqrt{\frac{y_{4}^2 + 2Yy_{3}}{Y^2} + y_{1}^2} \\
	x_{1} &= \arcsin\left(\frac{y_{2}}{Y_iy_{1}x_{3}}\right).
	\end{align*}
\end{proposition}

This result is essential for the solution of the state observation problem, which now reduces to the observation of $x_2$. We will show now that, unfortunately, this fundamental property of the state-to-output map\footnote{We mean here the mapping $(x_1,x_3,x_4) \mapsto y$.} is lost for the 4th order model \eqref{y}. To prove this fact we first observe that $y_1$ and $y_2$ are {\em unrelated} with $x$. Therefore, we propose the following definitions
$$
z:=\begmat{x_q'^2y^2_4 -y_2^2\\ {x_q'y_5 \over y_2} \\ {x_q'y_6 \over y_2}+y_2} \in \calb \subset \rea^3,\; v:=\begmat{x_1-y_1 \\ x_3 \\x_4} \in \cala \subset \rea^3,
$$
and look at the parameterized (in $y_2$) mapping $M_{y_2}:\cala \mapsto \calb$
$$
z =M_{y_2}(v)=\begmat{v^2_2+v^2_3-2y_2[v_2 \cos(v_1)+v_3 \sin (v_1)]\\v_2  \sin(v_1) -v_3 \cos(v_1)  \\ v_2  \cos(v_1) +v_3 \sin(v_1)}
$$

The task is to check whether $M_{y_2}$ is {\em injective}. This is equivalent to the existence of a mapping $M_{y_2}^I:\calb \mapsto \cala$ such that
$$
M_{y_2}^I(M_{y_2}(v))=v.
$$
We make now the second observation, namely, that $v^2_2+v^2_3$ is measurable. Hence, we can define  an alternative mapping $N_{y_2}:v \to z_N$
\begin{align}
\nonumber
z_N:=z-\begmat{v^2_2+v^2_3 \\ 0 \\ 0} =N_{y_2}(v)\\
\nonumber
=\begmat{-2y_2[v_2 \cos(v_1)+v_3 \sin (v_1)]\\v_2  \sin(v_1) -v_3 \cos(v_1)  \\ v_2  \cos(v_1) +v_3 \sin(v_1)}
\end{align}

We are in position to prove our claim of non-injectivity of the state-to-output map \eqref{y}.
\begin{lemma}\em
	\lab{lem1}
	The mapping  $(x_1,x_3,x_4) \mapsto (y_4,y_5,y_6)$ is non-injective.
\end{lemma}
\begin{proof} 
	In view of the discussion above, to prove the claim it suffices to show that
	$$
	\det\{\nabla N_{y_2}(v)\}=0.
	$$
	Some simple calculations yield
	
	\begin{align}
	\nonumber
	\nabla N_{y_2}(v)=\begmat{n_{11}(v) & -2y_2 \cos(v_1) & -2y_2 \sin(v_1)\\
		n_{21}(v) &   \sin(v_1) & -\cos(v_1) \\
		n_{31}(v)  &   \cos(v_1) & \sin(v_1)},
	\end{align}
	where 
	\begin{align}
	\nonumber
	n_{11}(v)&=2y_2[v_2 \sin(v_1)-v_3 \cos (v_1)],\\
	\nonumber
	n_{21}(v)&=v_2  \cos(v_1) +v_3 \sin(v_1),\\
	\nonumber
	n_{31}(v)&=-v_2  \sin(v_1) +v_3 \cos(v_1).
	\end{align}

		It is easy to see that
	$$
	\nabla N_{y_2}(v)\begmat{1 \\ -v_3 \\ v_2}=0,
	$$
	completing the proof.
\end{proof}

\section{Reparameterization of the Electrical Dynamics}
\label{sec4}
%
In this section we  propose a reparameterization of the electrical dynamics which is {\em linear} in $x_3$ and $x_4$. Moreover, we define three state-to-output mappings that will be used for the observer design.

\begin{lemma}\em
	\lab{lem2}
	Consider the SMIB model \eqref{x} with PMU measurements \eqref{y} and \eqref{y456}. 
	\begenu[{\bf C1}]
	\item There exists a matrix of measurable signals $\calm \in \rea^{2 \times 2}$ such that
	\begali{
		\lab{dotx3x4}
		\begmat{ \dot x_3 \\ \dot x_4}=\calm\begmat{x_3 \\ x_4} + \begmat{ 0 \\ c_1 u_2}.
	}
	\item There exists a measurable signal $Y =
	\begin{bmatrix}
	Y_1 \\ Y_2\\ Y_3
	\end{bmatrix}
	\in \rea^3$ such that
	\begequ
	\lab{Y}
	Y=\begmat{x^2_3+x^2_4 \\ e^{Jx_1} \begmat{x_3 \\ x_4}},
	\endequ
	with $J:=\begmat{0 & -1 \\ 1 & 0}.$
	\endenu
\end{lemma}

\begin{proof}
	To avoid cluttering, let us define the measurable signal
	\begequ
	\lab{z0}
	z_0:=y_6+ {y_2^2 \over x_q'}.
	\endequ
	From \eqref{y5} and \eqref{y6} we get after some simple calculations
	\begali{
		\nonumber
		y_5 x_4 + z_0 x_3 &= { y_2  \over x_q'} (x_3^2+x_4^2)\sin(x_1 - y_1)\\
		\lab{sincos}
		z_0 x_4 - y_5 x_3 &= {y_2 \over x_q'}  (x_3^2+x_4^2) \cos(x_1 - y_1).
	}
	Note also that from \eqref{y456} we have
	\begali{
		x^2_3+x^2_4&=(x_q')^2 y_4^2+2 x_q' y_6+ y_2^2 =Y_1,
		\lab{Y1}
	}
	where we notice that $Y_1$ is bounded away from zero. Replacing \eqref{Y1} in \eqref{sincos} and rearranging terms we get
	\begalis{
		y_2 \sin(x_1 - y_1)&= { x_q' \over Y_1}(y_5 x_4 + z_0 x_3)\\
		y_2 \cos(x_1 - y_1)&= { x_q' \over Y_1}(z_0 x_3 - y_5 x_3).
	}
	Replacing the latter equations in \eqref{x3} and \eqref{x4} and defining the matrix
	\begequ
	\lab{calm}
	\calm:=\begmat{ -a_2 + {b_2 x_q' \over Y_1} z_0 & {b_2  x_q' \over Y_1}y_5  \\
		- {b_1 x_q' \over Y_1} y_5 & -a_1 + {b_1 x_q' \over Y_1} z_0}.
	\endequ
	completes the proof of claim {\bf C1}.
	
	We proceed now to complete the proof of the claim {\bf C2}---notice that \eqref{Y1} is the first identity in \eqref{Y}. From \eqref{y5} and \eqref{y6} we get
\begin{align}
\nonumber
\begmat{-{x_q' \over y_2} y_5 \\   {x_q' \over y_2}y_6+y_2}&=\begmat{ \cos(x_1-y_1) & - \sin(x_1-y_1)  \\ \sin(x_1-y_1)&\cos(x_1-y_1)}\begmat{x_3 \\ x_4}\\
\nonumber
&=e^{J(x_1 - y_1)}\begmat{x_3 \\ x_4}.
\end{align}

	The proof is completed defining
	\begequ
	\lab{Y2Y3}
	\begmat{Y_2 \\ Y_3}:=e^{Jy_1} \begmat{-{ x_q' \over y_2}y_5 \\ { x_q' \over y_2}y_6 +y_2}.
	\endequ
\end{proof}

\begrem
\lab{rem3}
From \eqref{dotx3x4} and the first component in \eqref{Y}, that is $Y_1=x_3^2+x_4^2$, we see that we are dealing with a subsystem with linear dynamics but {\em nonlinear} state-output map. To the best of our knowledge \cite[Chapter 5]{ASTetal}, \cite[Chapter 3]{BER} there is no systematic way to design state observers for this class of systems. In Subsection \ref{subsec52} we present a gradient-descent based scheme \cite{SHI} for which some local stability properties can be established.
\endrem
%

\section{Proposed State Observer}
\label{sec5}
%
A corollary of the claim {\bf C2} is that there are two possibilities to reconstruct the states $(x_1,x_3,x_4)$. The first option is to find an observer for $x_1$ and get $(x_3,x_4)$ from the components $Y_2$ and $Y_3$ of \eqref{Y}. Alternatively, we can estimate $(x_3,x_4)$ and---as shown in the proposition below---obtain $x_1$ from simple trigonometric relations. The state $x_2$ can be reconstructed with the I\&I observer of \cite[Lemma 1]{LORetal}. 

Although the first approach looks simpler, the design of an observer for $x_1$ is still an open problem. Therefore, in this section we take the second route and design a GPEBO observer \cite{ORTetalgpebo} for the states  $(x_3,x_4)$. To enhance readability we divide the presentation of the observer in two parts, first---in the spirit of PEBO \cite{ORTetalscl} that translates the problem of state observation into one of parameter estimation---the derivation of a {\em nonlinear regression equation} required for the parameter estimation is given. Then, we invoke DREM \cite{ARAetaltac} to carry out the latter task with weak excitation requirements.
\subsection{Derivation of the regression equation for parameter estimation}
\label{subsec51}
%
\begin{lemma}
	\lab{lem3}\em
	Consider the electrical dynamics \eqref{dotx3x4} and the output map \eqref{Y}. Define the dynamic extension
	\begsubequ
	\lab{dotxiphi}
	\begali{
		\lab{dotxi}
		\dot{\xi} & =A(t) \xi + \begmat{ 0 \\ c_1 u_2} \\
		\lab{dotphi}
		\dot \Phi & = A(t) \Phi,\; \Phi(0)=I_2,
	}
	\endsubequ
	where we defined the time-varying matrix\footnote{That is, the evaluation of the matrix $\calm$, given in \eqref{calm}, along the trajectories of the system outputs.}
	\begequ
	\lab{a}
	A(t):=\begmat{ -a_2 + {b_2 x_q' \over Y_1(t)} z_0(t) & {b_2  x_q' \over Y_1(t)}y_5(t)  \\
		- {b_1 x_q' \over Y_1(t)} y_5(t) & -a_1 + {b_1 x_q' \over Y_1(t)} z_0(t)}.
	\endequ
	\begenu[{\bf P1}]
	\item There exists a constant vector $\theta \in \rea^{2}$ such that
	\begequ
	\lab{xxithe}
	\begmat{x_3 \\ x_4}= \xi+\Phi \theta.
	\endequ
	\item There exists measurable signals $y_E \in \rea$ and $\psi \in \rea^5$ such that
	\begin{align}
	\label{ye}
	y_E=\psi^\top  \Theta,
	\end{align}
	where we defined the constant vector
	\begequ
	\lab{The}
	\Theta:=\col(\theta_1,\theta_2,\theta_1 \theta_2,\theta_1^2,\theta_2^2).
	\endequ
	\item  Define the observer {
	\begequ
	\lab{staobs}
	\begmat{\hat x_1 \\ \hat x_3 \\ \hat x_4}= \begmat{\arcsin\Big({1 \over Y_1}\begmat{Y_3 & -Y_2}( \xi+\Phi \hat \theta)\Big) \\ \xi+\Phi \hat \theta},
	\endequ  }
	where  $\hat \theta$ is an estimate of the parameter $\theta$.  The following implication is true
	\begequ
	\lab{concon}
	\liminf \tilde \theta(t)=0\;\Rightarrow\;\liminf \begmat{\tilde x_1(t) \\ \tilde x_3(t) \\ \tilde x_4(t)}=0.
	\endequ
	\endenu
\end{lemma}

\begin{proof}
	Define the error signal
	\begequ
	\lab{e}
	\epsilon:= \begmat{x_3 \\ x_4}-\xi 
	\endequ
	and taking into account \eqref{dotx3x4}, \eqref{calm},  \eqref{dotxi} and \eqref{a} we obtain an LTV system $\dot \epsilon = A(t) \epsilon$. Now, from \eqref{dotphi} we see that $\Phi$ is the {\em state transition matrix} of the $\epsilon$ system. Consequently, there exists a {\em constant} vector $\theta \in \rea^2$ such that
	$$
	\epsilon = \Phi \theta,
	$$ 
	namely $\theta=\epsilon(0)$. We now have the following chain of implications
	\begalis{
		\epsilon = \Phi \theta & \Leftrightarrow \; \begmat{x_3 \\ x_4}=\xi+ \Phi \theta \qquad \qquad \qquad (\Leftarrow \eqref{e}) \\
		& \Rightarrow \; x^2_3 + x^2_4 =|\xi+ \Phi \theta   |^2   \quad \qquad \quad (\Leftarrow   |\cdot|^2)\\
		& \Leftrightarrow \; Y_1 = |\xi+ \Phi \theta   |^2  \qquad \qquad \qquad\;(\Leftarrow \eqref{Y}).
	}
	Notice that the right hand side of the first equivalence above proves property {\bf P1}. The proof of the property {\bf P2} follows developing the right hand side square above, rearranging terms and defining
	
	\begin{align}
	\nonumber
	y_E&:= Y_1 - |\xi|^2\\
	\nonumber
	\psi&:=
	\begin{bmatrix}
	2(\Phi_{11}\xi_1+\Phi_{21}\xi_2)\\
	2(\Phi_{12}\xi_1+\Phi_{22}\xi_2)\\
	2(\Phi_{11}\Phi_{12}+\Phi_{21}\Phi_{22})\\
	\Phi^2_{11}+\Phi^2_{21}\\
	\Phi^2_{12}+\Phi^2_{22}
	\end{bmatrix}.
	\end{align}
		
	The proof of the implication for the errors $\col(\tilde x_3, \tilde x_4)$ is obvious from \eqref{xxithe} and the definition of $\col(\hat x_3, \hat x_4)$ in \eqref{staobs}. To prove the claim for $\tilde x_1$ notice that using \eqref{Y} and the definition of $e^{Jx_1}$, we get
	\begalis{
		x_3Y_3-x_4Y_2 & =(x_3^2+x_4^2)\sin(x_1)\\
		& =Y_1\sin(x_1),
	}
	from which we obtain
	$$
	x_1=\arcsin\Big({x_3Y_3-x_4Y_2 \over Y_1}\Big).
	$$ 
\end{proof}

\subsection{DREM parameter estimator}
\lab{subsec52}
%
In view of the implication \eqref{concon} the remaining task to complete the observer design is to generate a {consistent} estimate for $\theta$. Towards this end, we dispose of the regressor equation \eqref{ye} that, unfortunately, is {\em nonlinear} in the unknown parameters $\theta$. Treating $\Theta$ as the unknown vector, it is possible to obtain an  {\em overparameterized linear} regression to which we can directly apply a classical gradient descent estimator, that is
\begequ
\lab{dotThe}
\dot {\hat \Theta}=-\Gamma \psi(\psi^\top   \hat \Theta - y_E),\;\Gamma>0.
\endequ
However, this approach has the following fundamental shortcoming. It is well-known  \cite[Theorem 2.5.1]{SASBOD} that a necessary and sufficient conditions for global (exponential) convergence of the gradient estimator is that the regressor $\psi$ satisfies a stringent {\em persistent excitation} requirement  \cite[Equation 2.5.3]{SASBOD}, which is not possible to satisfy in normal operation of the power system because of the overparameterization. To avoid this difficulty we propose here to use a DREM estimator that has the unique feature of generating $5$ new, one-dimensional linear regression equations to  {\em independently} estimate each of the parameters. This feature allows, on one hand, to estimate {\em only} the parameters $\theta$ and, on the other hand, to relax the excitation assumptions that guarantee its convergence.  This fact is illustrated in the simulations of Section \ref{sec5}. For further details on DREM the interested reader is refered to \cite{ARAetaltac}.

The first step to apply DREM to \eqref{ye} is to introduce a {\em linear, single-input $5$-output, bounded-input bounded-output (BIBO)--stable} operator $\calh$ and define the vector $Y_E \in \rea^5$ and the matrix $\Psi \in \rea^{q \times q}$ as
\begali{
	\nonumber
	Y_E & := \calh[y_E]\\
	\lab{YPsi}
	\Psi & :=\calh[\psi^\top].
}
Clearly, because of linearity and BIBO stability, these signals satisfy
\begequ
\label{extlre}
Y_E = \Psi \Theta.
\endequ 
At this point the key step of regressor ``mixing" of the DREM procedure is used to obtain a set of $5$ {\em scalar} equations as follows. First, recall that, for any (possibly singular) $q \times q$ matrix $M$ we have $\adj\{M\} M=\det\{M\}I_q$, where $\adj\{\cdot\}$ is the adjunct (also called ``adjugate") matrix. Now, multiplying from the left the vector equation \eqref{extlre} by the {\em adjunct matrix} of $\Psi$, we get
\begequ
\label{scalre}
\caly_i = \Delta \Theta_i,\quad i \in \{1,2,\dots,5\}
\endequ
where we have defined the signals
\begali{
	\nonumber
	\Delta & :=\det \{\Psi\} \in \rea\\
	\lab{delcaly}
	\caly &:= \adj\{\Psi\} Y_E \in \rea^5.
}

The availability of the scalar LREs \eqref{scalre} is the main feature of DREM that distinguishes it with respect to all other estimators and allows us to obtain significantly stronger results. Indeed, in DREM we propose the gradient-descent estimators\footnote{In the sequel, the quantifier $i \in \{1,2,\dots,5\}$ is omitted for brevity.} 
$$
\dot{\hat{\Theta}}_i = \gamma_i\Delta(\caly_i - \Delta \hat\Theta_i),
$$
with $\gamma_i >0$, which gives rise to the scalar error equations
$$
\dot {\tilde \Theta}_i=-\gamma_i \Delta^2 \tilde\Theta_i,
$$
whose explicit solution is 
\begali{
	\lab{solctpee}
	\tilde \Theta_i(t) = e^{-\gamma_i \int_0^t \Delta^2(s)ds} \tilde \Theta_i(0).
}
From direct inspection of \eqref{solctpee} we conclude that following {\em equivalence} holds
$$
\liminf \tilde \Theta_i(t)=0~~\Leftrightarrow~~\Delta(t) \notin \call_2,
$$ 
and convergence can be made {\em arbitrarily fast} simply increasing $\gamma_i$. 

\begrem
\lab{rem4}
It is clear from the definition of $\Theta$ in \eqref{The} that we are only interested in the first and second components of this vector.
\endrem
\begrem
\lab{rem17}
In \cite{Coretall} it was observed that, using Cramer$'$s  rule, the computation of the adjunct matrix $adj\{\Psi^T\}$ can be avoided. Indeed, the elements of the vector $Y$ can be computed as
\begin{align}
\caly _i= \det \{\Psi_{Y_i}^T \}
\end{align}
where the matrix $\Psi_{Y_i}$ is obtained replacing the $i$-th row of $\Psi$ by the vector $Y$ .
\endrem

\subsection{Main stability result}
\lab{subsec53}
%
We are now in position to present the main result of this paper, a {\em globally convergent} observer for the state of the SMIB power system \eqref{x} with measurable outputs \eqref{y}, with the required excitation conditions been rather weak. As indicated before, the state $x_2$ can be reconstructed with the I\&I observer of \cite[Lemma 1]{LORetal} and is omitted for brevity. 

\begin{proposition}\em
	\lab{pro2}
	Consider the SMIB power system (\ref{x}), \eqref{y} verifying Assumption \ref{ass1}. Fix a stable transfer matrix\footnote{The latter condition on the constants $d_i$ is necessary to avoid the possibility of $\Psi$ been singular.} 
	\begequ
	\lab{calhs}
	\calh(s)=\begmat{1\\{d_2 \over s+d_2} \\ \vdots \\ {d_5 \over s+d_5}},\;d_i>0,\;d_i \neq d_j,\;\forall i \neq j.
	\endequ
	Let the state observer be defined by  \eqref{Y}, \eqref{z0}, \eqref{dotxiphi}, \eqref{a}, \eqref{staobs}, \eqref{e}, \eqref{YPsi}, \eqref{delcaly} together with the parameter estimators
	\begali{
		\dot {\hat \theta}_k &=-\gamma_k \Delta (\Delta \hat{\theta}_k-\caly_k),\;\gamma_k>0,\;k=1,2.
		\lab{paa}
	} 
	If $\Delta \notin \call_2$ then
	$$
	\liminf \begmat{\tilde x_1(t) \\ \tilde x_3(t) \\ \tilde x_4(t)}=0,
	$$
	with all signals bounded.
\end{proposition}

\begin{proof}
	Given the derivations of Subsection \ref{subsec52} we get the parameter estimator error equations
	$$
	\dot {\tilde \theta}_k =-\gamma_k \Delta^2 \tilde \theta_k,\;k=1,2.
	$$
	Clearly, with the standing assumption on $\Delta$, we have that $\tilde \theta(t) \to 0$. The proof is completed invoking \eqref{concon}.
	
\end{proof}

\begrem
\lab{rem5}
Although the construction of DREM allows for the use of general, LTV, BIBO-stable operators $\calh$, for the sake of simplicity we consider here the use of simple LTI filters. Moreover, we take the first element of the matrix to be the identity. 
\endrem
%
\section{Simulation Results}
\lab{sec6}
In this section we present some simulations that illustrate the performance of the observer of the states $(x_3,x_4)$ of Proposition \ref{pro1}, which combines GPEBO with DREM. For the sake of comparison we also show the simulation results of GPEBO with the overparameterized parameter estimator \eqref{dotThe} and a simple {\em state} observer directly derived from optimization considerations.

\subsection{Single Machine Infinite Bus}
We simulated the system \eqref{x} with the parameters 
$a_0=13.2893$, $a_1= 0.268$,$a_2=7.7462$, $b_0=6.6447$,  $b_1=0.1564$, $b_2=4.5204$, $c_1=0.1116$,
obtained from Table \ref{tab:params} with $u_1=u_2=0.1$. The systems initial conditions were set to $x(0)=\col(0.1,0.2,0.4,0.3)$. The initial conditions of {\em all} the states of the three observers were set to zero. 

\begin{table}
	\centering
	\begin{tabular}{|l|l|l|l|}
		\hline
		Symbol & Description& Value & Unit \\  	\hline
		$D$ &Damping factor& 2 & pu\\
		$H$ &Inertia constant & 23.64 &sec \\
		$k$ & Tuning parameter & 80 & -\\
		$T_{d0}'$&Direct-axis transient  &  8.96&sec\\
		&open-circuit time constant&&\\
		$x_d$ &Direct-axis reactance  & 0.146&  pu\\
		$x_d'$&Direct-axis transient reactance& 0.0608& pu\\
		$Y$ &Stator admittance& 16.45& pu\\ 
		$\omega_s$&Nominal synchronous speed&314.16& rad/sec\\
		\hline 	
	\end{tabular}
	\caption{Parameters for the SMIB system \eqref{x}.}
	\label{tab:params}
\end{table} 

\subsubsection{GPEBO+DREM of Proposition \ref{pro1}}
\label{subsec61}

The parameters of the transfer matrix \eqref{calhs} were chosen as $d_2=2,\;d_3=4,\;d_4=6,\;d_5=8$. 

Figures  \ref{fig1} and \ref{fig2} show the transients of $x_3$ and $x_4$ and their observed values  $\hat x_3$ and $\hat x_4$ with the adaptation gains $\gamma_1=\gamma_2=:\gamma$ and different values for $\gamma$. As expected, increasing $\gamma$ speeds-up the convergence generating some mild overshoots in $\hat x_3$.

\begin{figure}[h]
	\centering
	\includegraphics[width=1\linewidth]{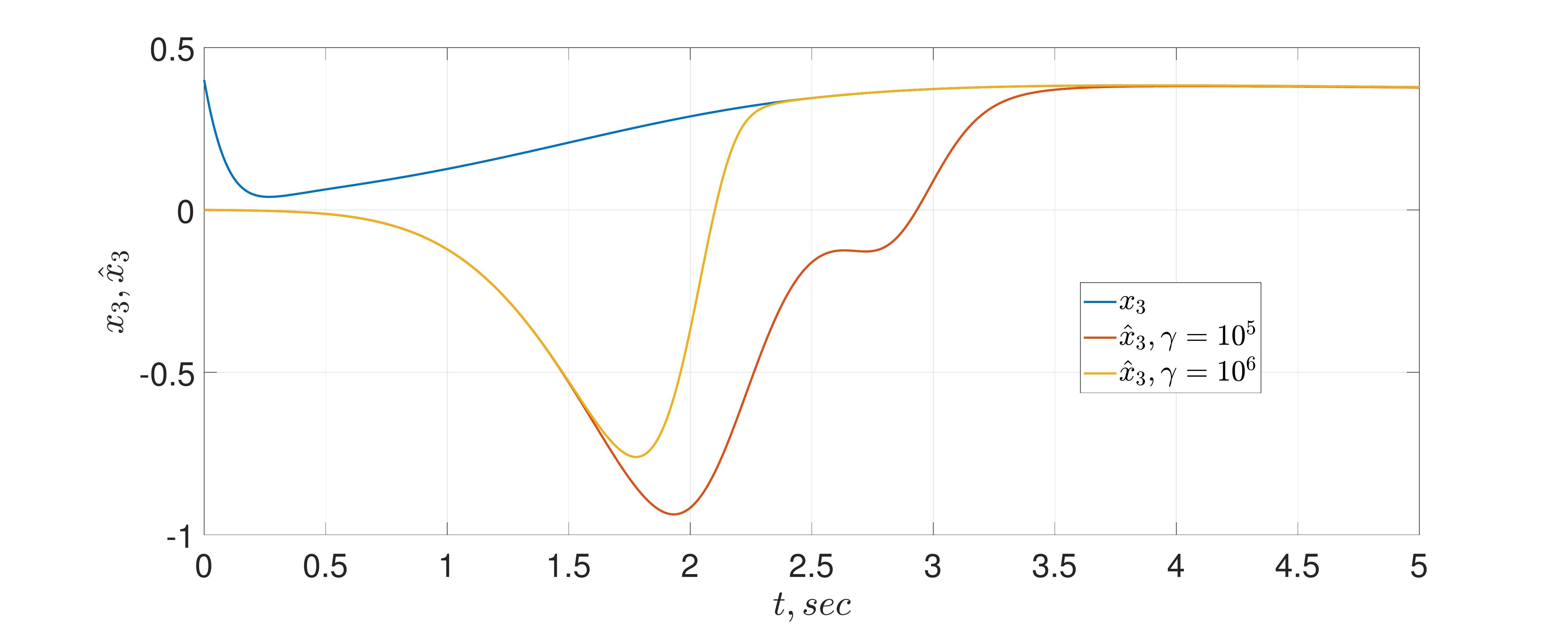}
	\caption{Transients of $x_3$ and $\hat{x}_3$ of GPEBO+DREM for different values of $\gamma$}
	\label{fig1}
\end{figure}

\begin{figure}[h]
	\centering
	\includegraphics[width=1\linewidth]{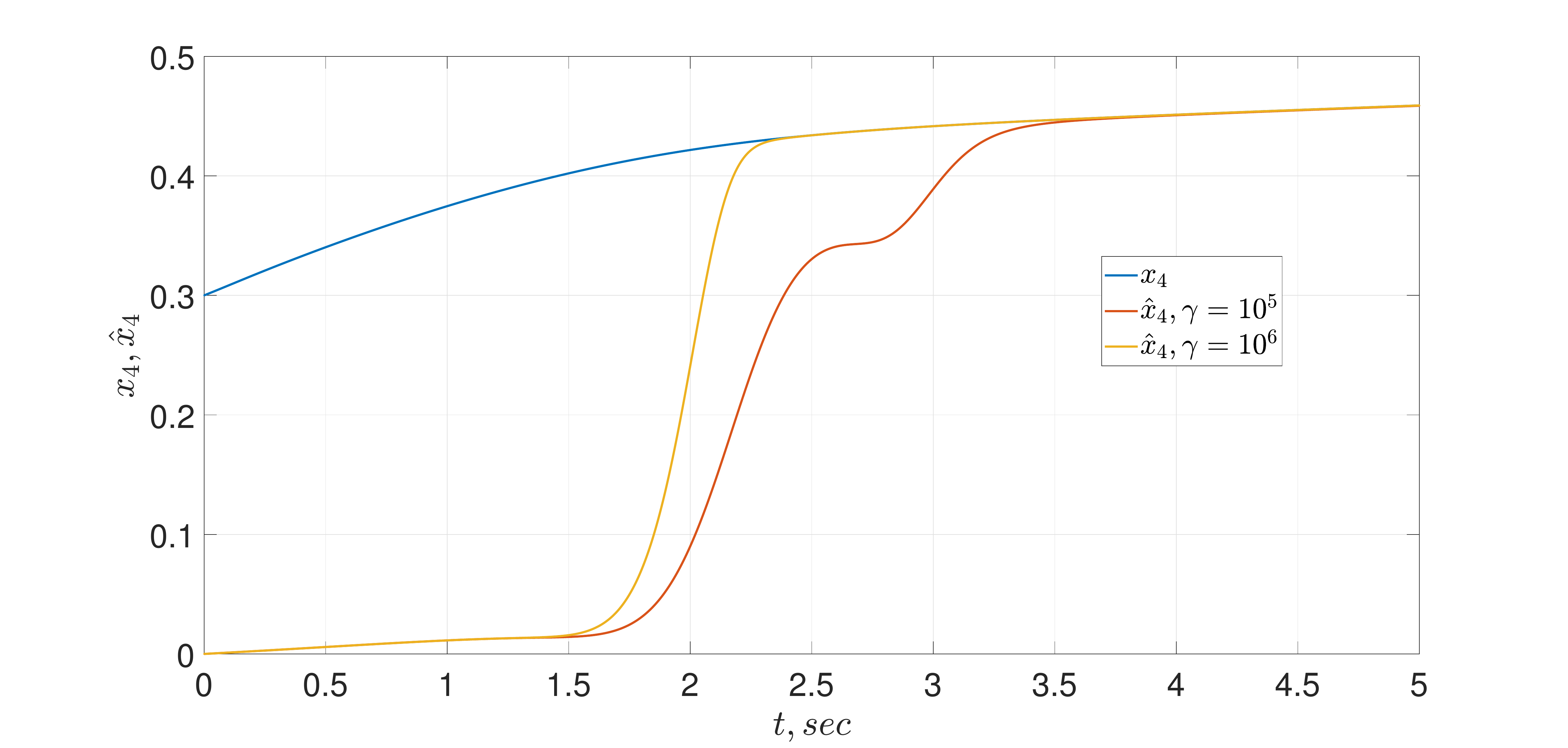}
	\caption{Transients of $x_4$ and $\hat{x}_4$ of GPEBO+DREM  for different values of $\gamma$}
	\label{fig2}
\end{figure}

\subsubsection{GPEBO with overparameterized estimator \eqref{dotThe}}
\label{subsec62}
%
In this subsection we show that the standard gradient estimator \eqref{dotThe} for the overparameterized regression is inadequate. Define the vector 
	\begequ
\lab{error}
	e:=\begmat{\Theta_1\Theta_2-\Theta_3 \\ \Theta_1-\Theta_4^2\\ \Theta_2-\Theta_5^2}.
\endequ
From the definition of the vector $\Theta$ in \eqref{The} we have that $e \equiv 0$. Figures  \ref{err1}  and \ref{err2} show the transients of the estimated vector $\hat{e}$ with the adaptation gains $\Gamma=10^6 I_5$ and  $\Gamma=10^8 I_5$, which does not converge to zero---proving that the parameters do not converge to their true values. Several different values of $\Gamma$ were tried, observing always an erroneous behavior.

\begin{figure}[h]
	\centering
	\includegraphics[width=1\linewidth]{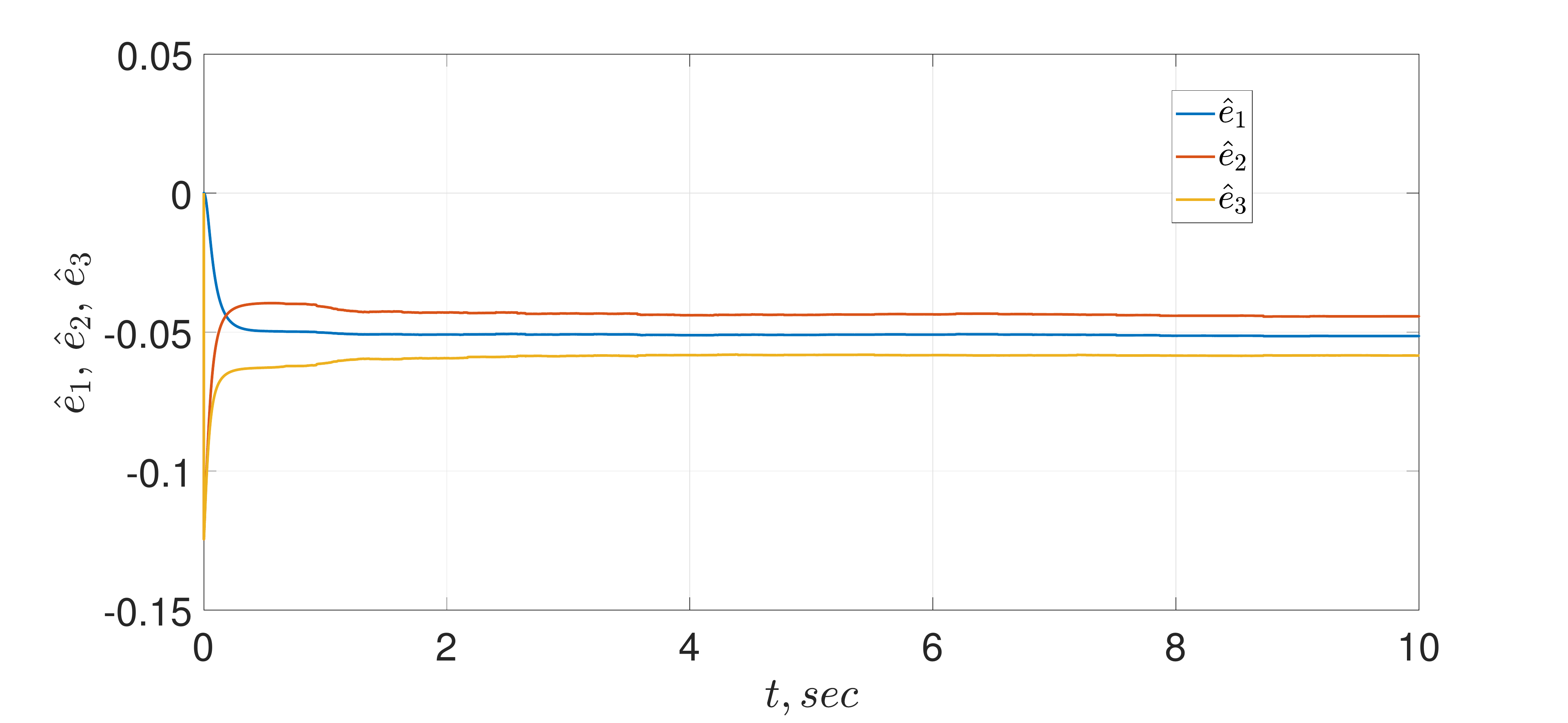}
	\caption{Transients of $\hat{e}$ for $\Gamma=10^6 I_5$ of the overparameterized estimator \eqref{dotThe}  }
	\label{err1}
\end{figure}

\begin{figure}[h]
	\centering
	\includegraphics[width=1\linewidth]{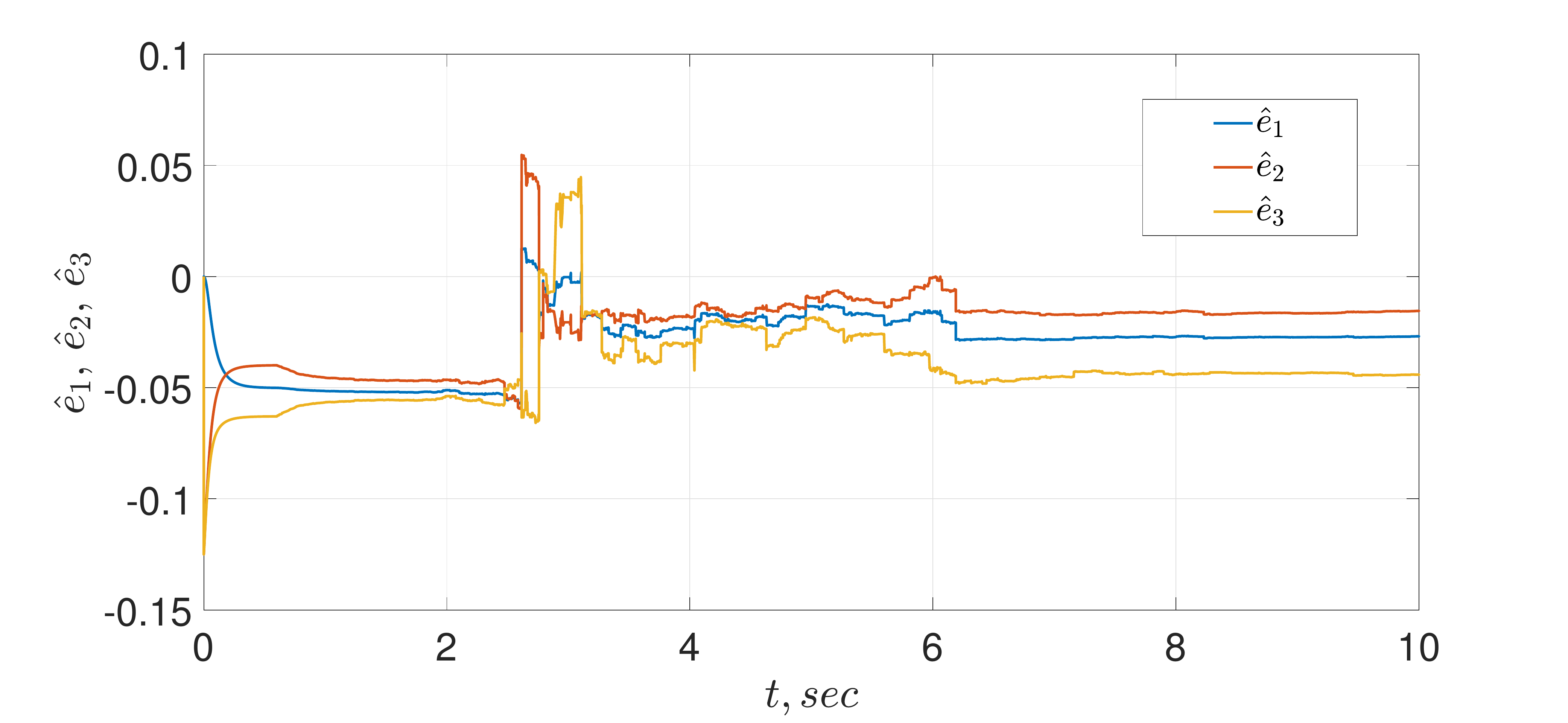}
	\caption{Transients of $\hat{e}$ for $\Gamma=10^8 I_5$ of the overparameterized estimator \eqref{dotThe} }
	\label{err2}
\end{figure}

\subsubsection{Gradient-descent state estimation algorithm}
\label{subsec63}
In this subsection we propose to design gradient-descent algorithms, directly for observation of the {\em states} $(x_3,x_4)$, proceeding from the state-to-output map \eqref{Y1}. The gradient descent-based approach to state observation was, apparently, first proposed in \cite{SHISUZNUK}, and has been pursued recently by several researchers \cite{BERPRA,FORetal,LAGetal,ORTetalcst,SHI}. 

The construction proceeds as follows. Given the criterion
$$
\calt(x_3,x_4):={1 \over 4}[Y_1-(x_3^2+x_4^2)]^2,
$$
with $Y_1$ given in \eqref{Y1}, propose an observer
$$
\begmat{\dot {\hat x}_3 \\ \dot {\hat x}_4}=-\Gamma \nabla \calt(\hat x_3,\hat x_4)+\Bigg(A(t)\begmat{\hat x_3 \\ \hat x_4} + \begmat{ 0 \\ c_1 u_2}\Bigg)
$$
with $\Gamma \in \rea^{2 \times 2}$ positive definite. That is,
	\begequ
\lab{gdst}
\begmat{\dot {\hat x}_3 \\ \dot {\hat x}_4}=\Gamma [Y_1-(\hat x_3^2+\hat x_4^2)]\begmat{\hat x_3 \\ \hat x_4}+\Bigg(A(t)\begmat{\hat x_3 \\ \hat x_4} + \begmat{ 0 \\ c_1 u_2}\Bigg)
\endequ
The local stability properties of this observer can be studied using the Taylor-expansion based analysis proposed in \cite{SHI}.

Figures  \ref{fig3} and \ref{fig4} show the transients of $x_3$ and $x_4$ and their observed values  $\hat x_3$ and $\hat x_4$ with  $\Gamma=\gamma I$ and different values of $\gamma$. Interestingly, the state estimation errors converge to zero, even for large initial conditions errors. However, the transient behavior is significantly slower that the one of GPEBO+DREM---notice the difference in time scales.  

\begin{figure}[h]
	\centering
	\includegraphics[width=1\linewidth]{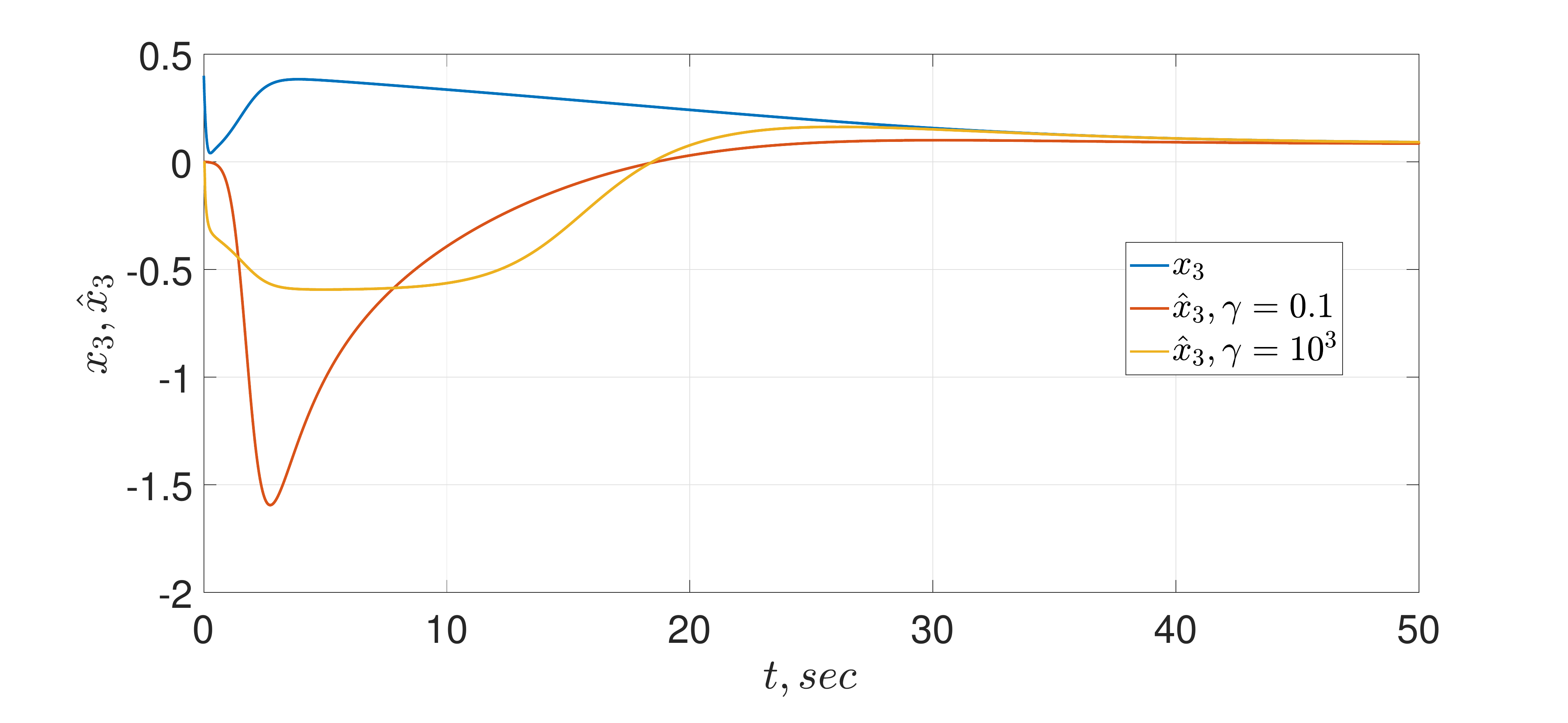}
	\caption{Transients of $x_3$ and $\hat{x}_3$ of the gradient descent observer for different values of $\Gamma$}
	\label{fig3}
\end{figure}

\begin{figure}[h]
	\centering
	\includegraphics[width=1\linewidth]{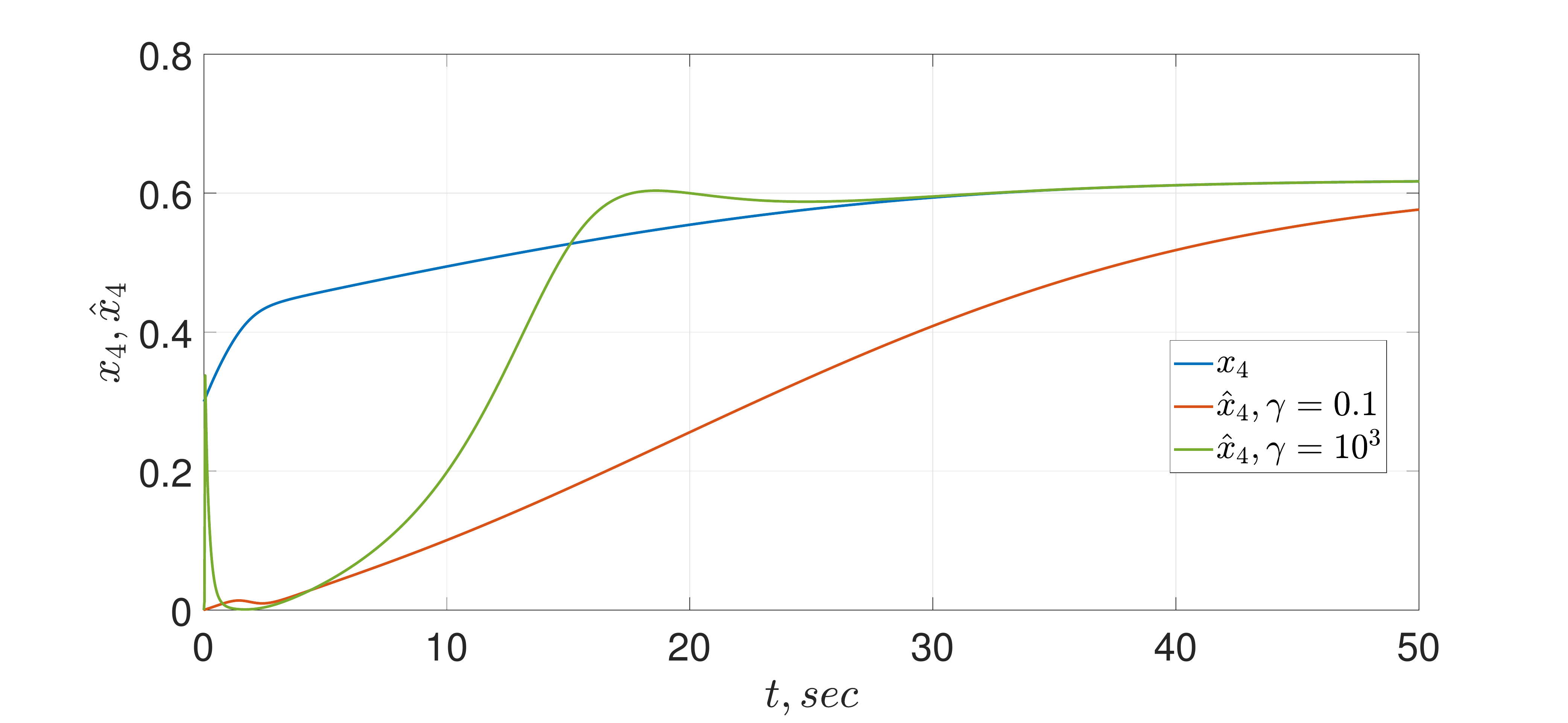}
	\caption{Transients of $x_4$ and $\hat{x}_4$ of the gradient descent observer for different values of $\Gamma$}
	\label{fig4}
\end{figure}
{
\subsection{Multi-Machine Power System}
We simulated the well-known New England IEEE 39 bus system shown in \autoref{fig:IEEE39Bus}, with the parameters provided in \cite{hiskens13}. All synchronous generators are represented by the fourth-order flux-decay model \eqref{x} and are equipped with automatic voltage regulators and power system stabilizers according to \cite{hiskens13}. 
To monitor the system, we assume that a PMU is installed at the terminal bus of generators 6. 

As a test case we used minor load variations in the system. The resulting frequency variations are within $60\pm0.020$~[Hz] and hence consistent with those during regular operation of transmission grids \cite{weissbach09}.

\begin{figure}
	\centering
	\includegraphics[width=1\linewidth]{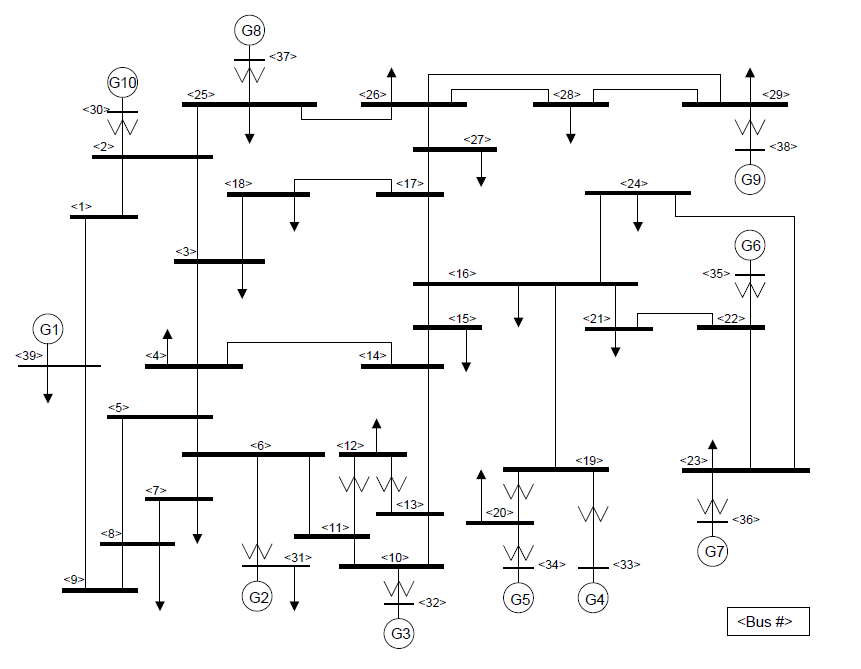}
	\caption{New England IEEE 39 bus system (figure taken from \cite{hiskens13}).}
	\label{fig:IEEE39Bus}
\end{figure}

\subsubsection{GPEBO+DREM and algebraic observer of Proposition \ref{pro2}}
The parameters of the transfer matrix \eqref{calhs} were chosen as $d_2=2,\;d_3=4,\;d_4=6,\;d_5=8$. Different values were chosen for the adaptation gains $\gamma_i=\gamma$, $i = 1, ... ,5 $. In \autoref{fig1_mm} the simulation results for $x_1, x_2, x_3 $ and the state estimation of the observer of Proposition \ref{pro2} are shown. As seen from the figure consistent estimation of the state variables is achieved.
\begin{figure}[h]
	\centering
	\includegraphics[width=.85\linewidth]{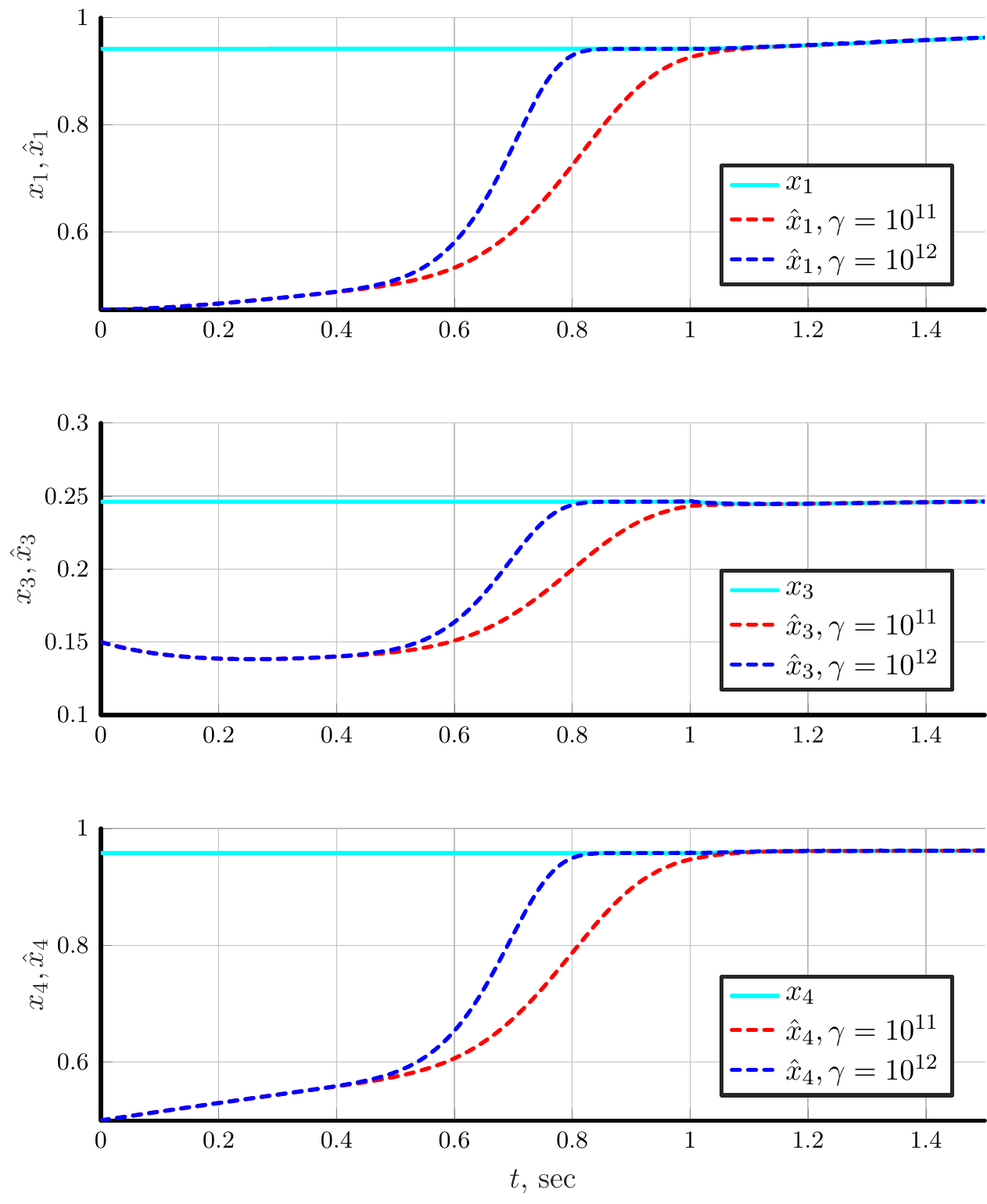}
	\caption{Transients of the GPEBO+DREM  observer, with different values of $\gamma$, for generator 6 in the presence of load variations.}
	\label{fig1_mm}
\end{figure}

\subsubsection{GPEBO with overparameterized estimator \eqref{dotThe}}
The overparameterized estimator was simulated using different values for the adaptation gain $\Gamma = \gamma I_5$. The elements of error vector defined in \eqref{error} are given in \autoref{fig2_mm}, showing that convergence is not achieved.
\begin{figure}[h]
	\centering
	\includegraphics[width=.85\linewidth]{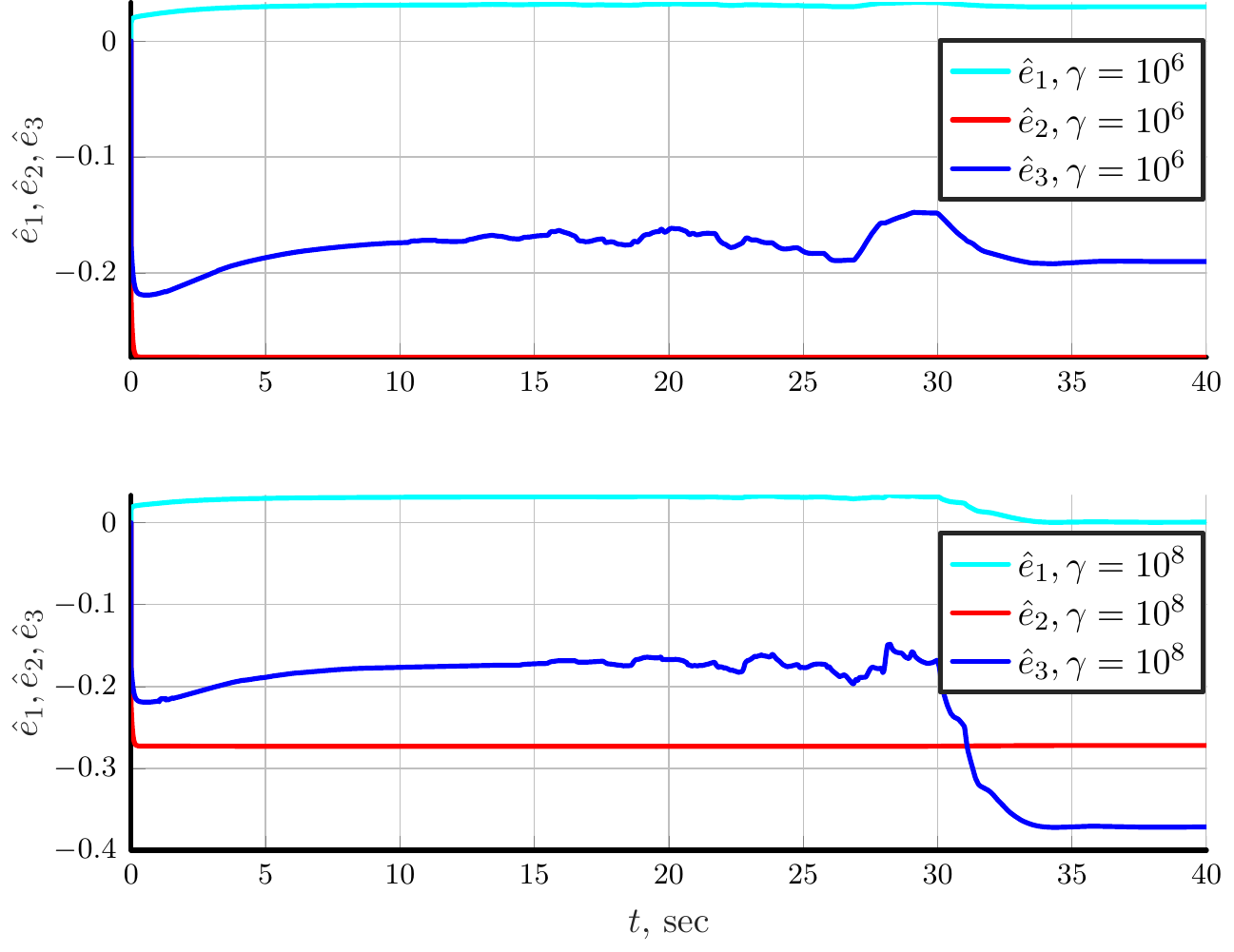}
	\caption{Transients of the error vector \eqref{error}  of the overparameterized estimator \eqref{dotThe} for generator 6 in the presence of load variations.}
	\label{fig2_mm}
\end{figure}

\subsubsection{Gradient-descent state estimation algorithm \eqref{gdst}}
In \autoref{fig3_mm} the simulation results for the gradient-descent state estimation algorithm introduced in \eqref{gdst} are shown using different values for the gain $\Gamma = \gamma I_2$. As seen from the figure the transient behavior is very good, mainly due to the rapid change of the state variables $x_3$ and $x_4$ that provide the required excitation to estimate the gradient. 
\begin{figure}[h]
	\centering
	\includegraphics[width=.85\linewidth]{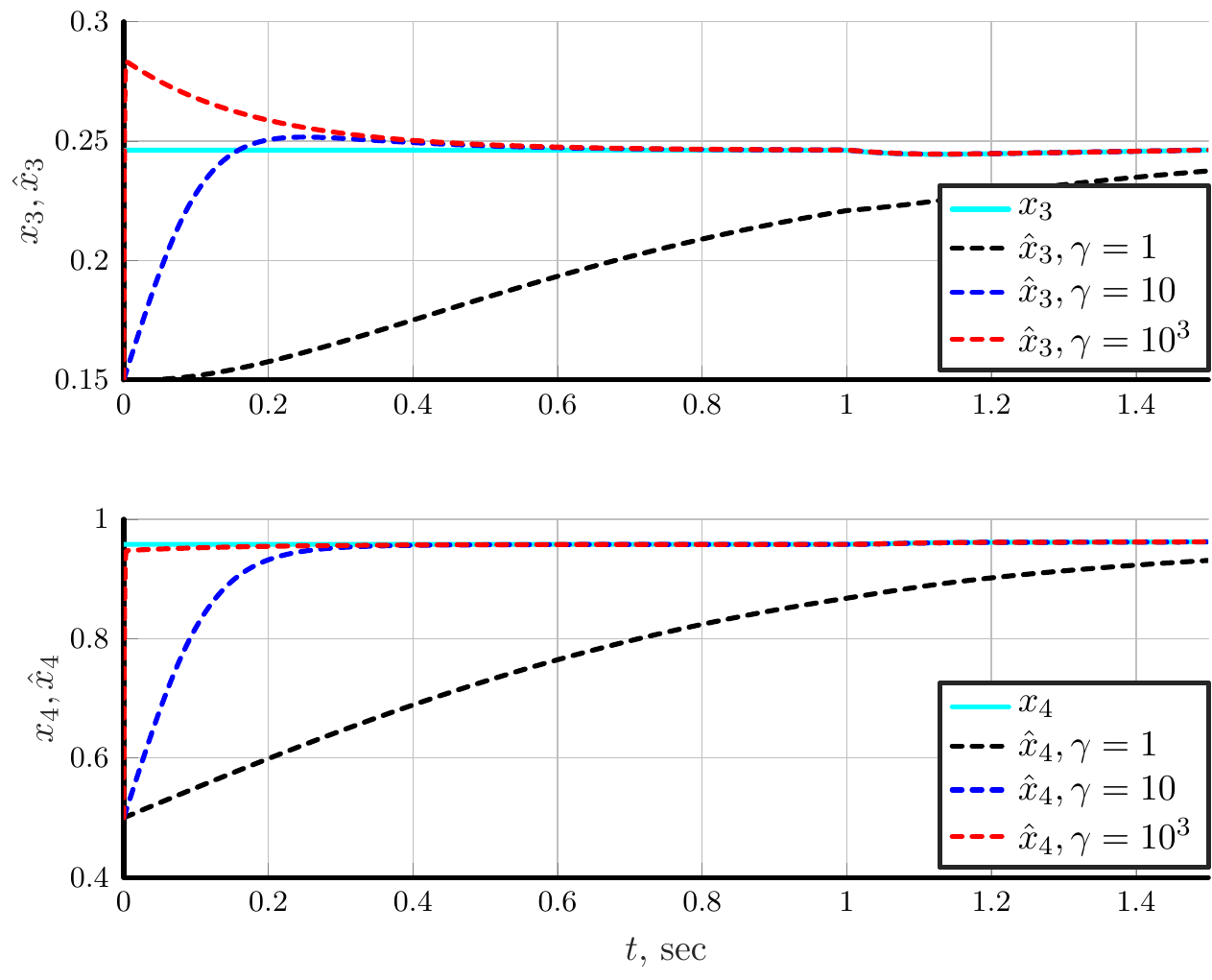}
	\caption{Transients of the gradient descent observer \eqref{gdst} of generator 6 in the presence of load variations.}
	\label{fig3_mm}
\end{figure}

}
\section{Conclusion and Future Research}
\lab{sec7}

We have proposed s globally convergent observer for the state estimation, from PMU measurements, of multimachine power systems described by the widely popular fourth order model \eqref{x}. It is shown that we can concentrate on the observation of the states $(x_3,x_4)$ and compute $x_1$ from an explicit algebraic equation. 

For the observation of $(x_3,x_4)$ we have also proposed a gradient-descent based observer that, in spite of the lack of a global convergence proof, performs quite well in a realistic multimachine scenario. A topic of current research is to assess the convergence properties of this observer---beyond the local analysis based on linearization of \cite{SHI}. 

Another interesting possibility motivated by  \eqref{Y} is to design a gradient-descent observer for  $(x_1,x_3,x_4)$ fixing a cost function
$$
\calt_N(x_1,x_3,x_4):=\Bigg|Y-\begmat{x^2_3+x^2_4 \\ e^{Jx_1} \begmat{x_3 \\ x_4}}\Bigg|^2,
$$ 
and going in the direction of descent of the gradient---with respect to $(x_1,x_3,x_4)$---of this cost function. We hope to be able to report this result in the near future.

\section*{Acknowledgment}

This paper is partially supported by the Ministry of Science and Higher Education of Russian Federation, passport of goszadanie no. 2019-0898, and by Government of Russian
Federation (Grant 08-08).

\ifCLASSOPTIONcaptionsoff
  \newpage
\fi

%

\begin{IEEEbiography}[{\includegraphics[width=1in,height=1.25in,clip,keepaspectratio]{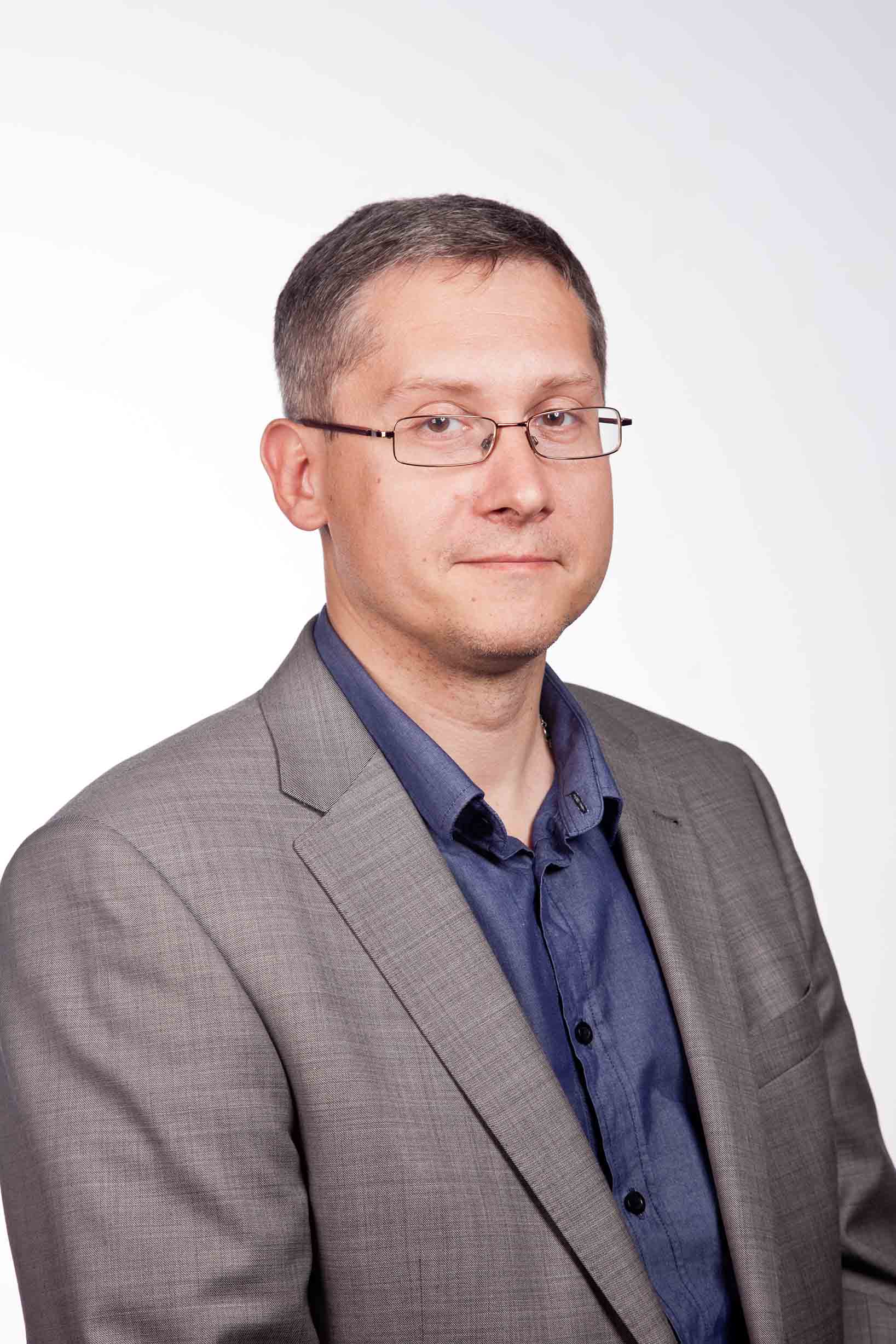}}]{Alexey Bobtsov}

(SM$'$10) received the M.S. degree in electrical engineering from ITMO University, St. Petersburg, Russia in 1996, received his PhD in 1999 and the degree of Doctor of Science (habilitation thesis) in 2007 from the same University.
From November 1999 to December 2000 he served
as Assistant Lecturer of Department of Automation
and Remote Control. From December 2000 to May
2007 Dr. Bobtsov served as Associate Professor of
Department of Control Systems and Informatics. In
May 2007 Dr. Bobtsov was appointed as Professor
of Department of Control Systems and Informatics. In September 2008 he was elected as the Dean of Computer Technologies and Control Faculty. He
is currently the Dean of School of Computer Science and Control at ITMO University. He is a Senior Member of IEEE since 2010. He is a Member of International Public Association Academy of Navigation and Motion Control.
He is coauthor of more than 300 journal and conference papers, 5 patents, 15 books and textbooks. His research interests are in fields of nonlinear and adaptive control, control of oscillatory and chaotic systems and computeraided control systems design with applications to mechanical and robotic systems.
\end{IEEEbiography}


\begin{IEEEbiography}[{\includegraphics[width=1in,height=1.25in,clip,keepaspectratio]{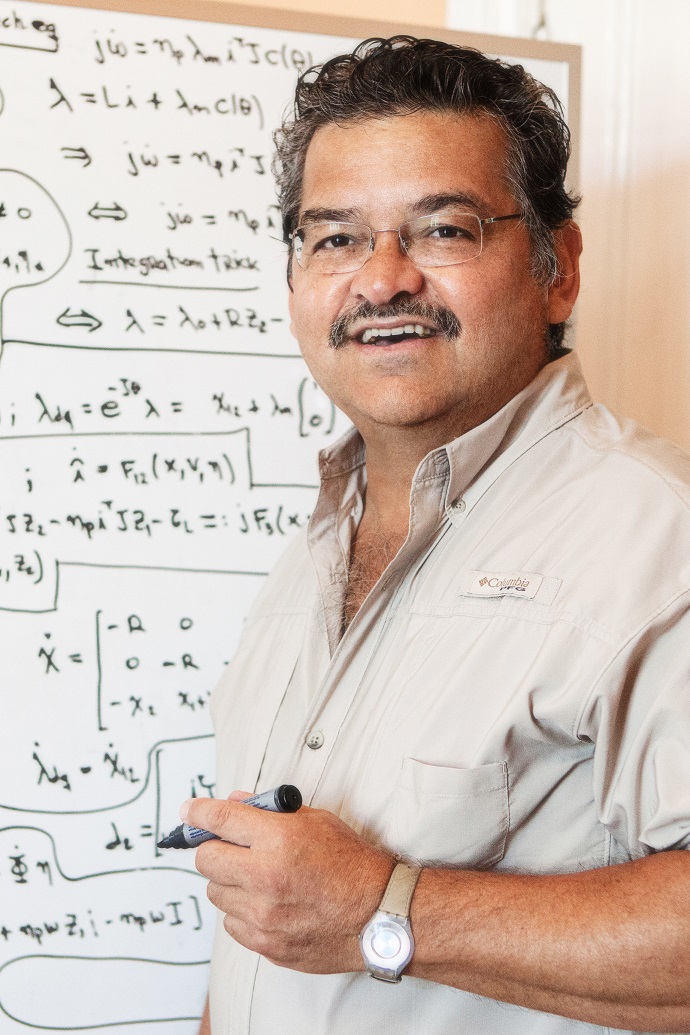}}]{Romeo Ortega}
	
(S$'$81, M$'$85, SM$'$98, F$'$99, LF$'20$) was
born in Mexico. He obtained his BSc in Electrical and Mechanical Engineering from the National
University of Mexico, Master of Engineering from
Polytechnical Institute of Leningrad, USSR, and the
Docteur D?Etat from the Polytechnical Institute of
Grenoble, France in 1974, 1978 and 1984 respectively.
He then joined the National University of Mexico,
where he worked until 1989. He was a Visiting
Professor at the University of Illinois in 1987-88 and
at the McGill University in 1991-1992, and a Fellow of the Japan Society for
Promotion of Science in 1990-1991. He has been a member of the French
National Researcher Council (CNRS) since June 1992. Currently he is in
the Laboratoire de Signaux et Systemes (CENTRALE-SUPELEC) in Gif-sur-Yvette. His
research interests are in the fields of nonlinear and adaptive control, with
special emphasis on applications.
Dr Ortega has published three books and more than 350 scientific papers
in international journals, with an h-index of 83. He has supervised 35 PhD
thesis. He has served as chairman in several IFAC and IEEE committees and
participated in various editorial boards. Currently, he is Editor in Chief of Int.
J. on Adaptive Control and Signal Processing.
\end{IEEEbiography}

\begin{IEEEbiography}[{\includegraphics[width=1in,height=1.25in,clip,keepaspectratio]{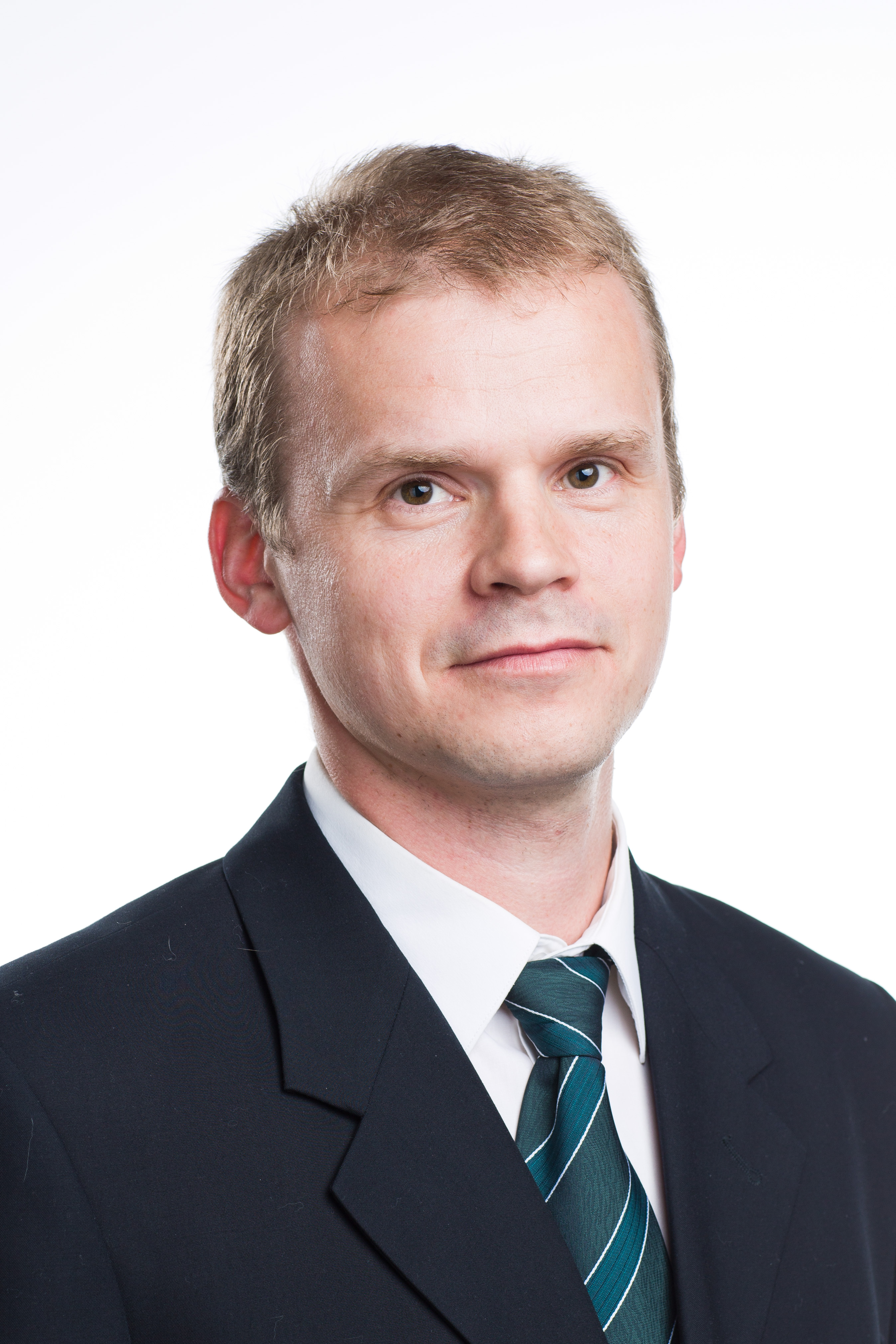}}]{Nikolay Nikolaev}

	received the M.S. degree in
	electrical engineering from ITMO University, St.
	Petersburg, Russia in 2003, received his PhD in 2006
	from the same University. From 2002 until 2013
	he worked as Engineer of Department of Control
	Systems Design for Power Plants at JSC Kirovsky
	Zavod (Kirov Plant). From 2013 he is an Assistant
	Professor in Department of Control Systems and
	Robotics from ITMO University. He is a Member
	of IEEE since 2006. His research interests are in
	fields of nonlinear and adaptive control.
\end{IEEEbiography}

\begin{IEEEbiography}[{\includegraphics[width=1in,height=1.25in,clip,keepaspectratio]{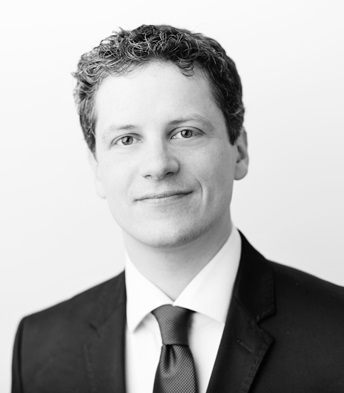}}]{Johannes Schiffer}
received the Diploma degree
in engineering cybernetics from the University of
Stuttgart, Germany, in 2009 and the Ph.D. degree
(Dr.-Ing.) in electrical engineering from Technische
Universitat (TU) Berlin, Germany, in 2015. ¨
He currently holds the chair of Control Systems
and Network Control Technology at Brandenburgische Technische Universitat Cottbus-Senftenberg, ¨
Germany, where he also serves as Deputy of Research. Prior to that, he has held appointments as
Lecturer (Assistant Professor) at the School of Electronic and Electrical Engineering, University of Leeds, U.K. and as Research
Associate in the Control Systems Group and at the Chair of Sustainable
Electric Networks and Sources of Energy both at TU Berlin.
In 2017 he and his co-workers received the Automatica Paper Prize over
the years 2014-2016. His current research interests include distributed control
and analysis of complex networks with application to microgrids and power
systems.
\end{IEEEbiography}

\begin{IEEEbiography}[{\includegraphics[width=1in,height=1.25in,clip,keepaspectratio]{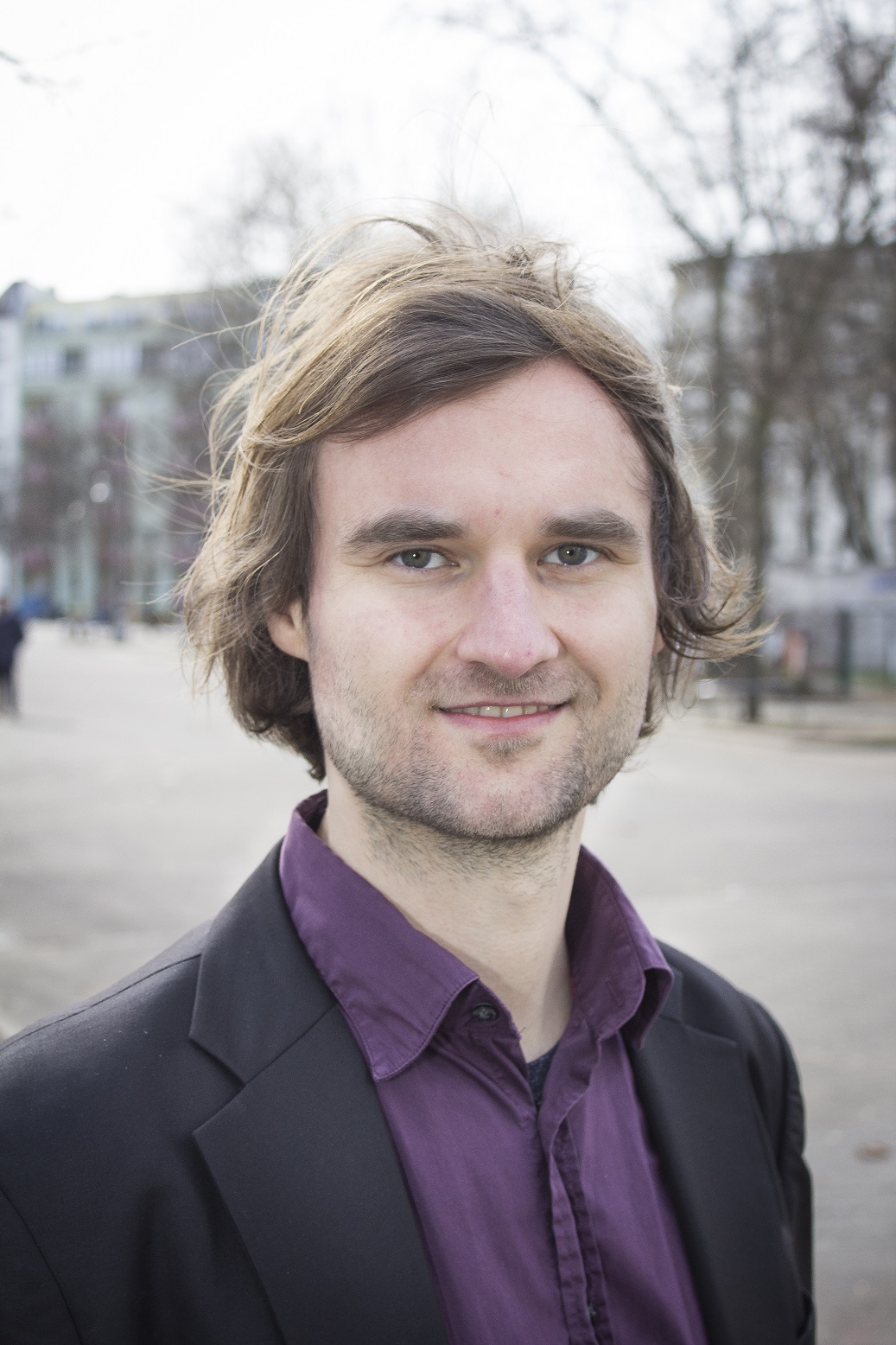}}]{M. Nicolai L. Lorenz-Meyer}
 received his M.Sc. in
Engineering Science from the Technical University
of Berlin in 2019 and his B.Eng. in Business Administration and Engineering for Environment and
Sustainability from the Beuth University of Applied
Sciences and the Berlin School of Economics and
Law in 2016.
He is currently pursuing the Ph.D. degree with the
the chair of Control Systems and Network Control
Technology at the Brandenburg University of Technology Cottbus-Senftenberg, Germany. His current
research interests include the development of control theory-based methods
for on-line dynamics security assessment in power systems.
\end{IEEEbiography}




\end{document}